\documentclass[journal]{IEEEtran}

\ifCLASSINFOpdf
\else
\fi

\usepackage{url}
\usepackage{amsmath,amssymb}
\usepackage{float}
\usepackage{graphicx,subfigure}
\usepackage{color}
\usepackage[colorlinks=true,linkcolor=black]{hyperref}
\usepackage{breakurl}
\usepackage{cite}
\usepackage[ruled,linesnumbered]{algorithm2e}
\usepackage{ amssymb }
\usepackage{mathrsfs }
\usepackage{enumitem} 
\usepackage{romannum}
\usepackage{hyperref}
\usepackage{amsthm}
\hypersetup{colorlinks,linkcolor={black},citecolor={black},urlcolor={red}}  

\newtheorem{theorem}{Theorem}
\newtheorem{lemma}{Lemma}

\newtheorem{corollary}{Corollary}
\newtheorem{problem}{Problem}

\begin{document}

\title{\LARGE\bf Strategies to Inject Spoofed Measurement Data\\ to Mislead Kalman Filter}%
\author{Zhongshun~Zhang, 
        Lifeng~Zhou,
        and~Pratap~Tokekar
\thanks{The authors are with the Department of Electrical \& Computer Engineering, Virginia Tech, USA. \texttt{\small \{zszhang, lfzhou, tokekar\}@vt.edu}.}%
\thanks{This material is based upon work supported by the National Science Foundation under Grant \#1566247.}
\thanks{A preliminary version of this paper, without the evaluation with $\chi^2$ detector in Section~\ref{sec:detector}, was presented at ACC 2018~\cite{zhang2018strategies}.}}

\maketitle
\thispagestyle{empty}
\pagestyle{empty}

\begin{abstract}

We study the problem of designing false measurement data that is injected to corrupt and mislead the output of a Kalman filter. Unlike existing works that focus on detection and filtering algorithms for the observer, we study the problem from the attacker's point-of-view. In our model, the attacker can corrupt the measurements by injecting additive spoofing signals. The attacker seeks to create a separation between the estimate of the Kalman filter with and without spoofed signals. We present a number of results on how to inject spoofing signals while minimizing the magnitude of the injected signals. The resulting strategies are evaluated through simulations along with theoretical proofs. We also evaluate the spoofing strategy in the presence of a $\chi^2$ spoof detector. The results show that the proposed strategy can successfully mislead a Kalman filter while ensuring it is not detected. 
\end{abstract}

\section{Introduction}
As autonomous systems proliferate, there are growing concerns about their security and safety~\cite{parkinson2017cyber,thing2016autonomous}. Of particular concern is their vulnerability to signal spoofing attacks~\cite{tippenhauer2011requirements}. As a result, many researchers are designing algorithms that enable an \emph{observer} to detect and mitigate signal spoofing attacks (e.g.,~\cite{al2015design,chen2007detecting,gil2017guaranteeing,zhang2017functional,fan2017synchrophasor}). We study the problem from the opposite (i.e., the attacker's) point-of-view. Our goal is to characterize the capabilities of the attacker that is generating the spoofing signals while assuming that the observer is using a Kalman filter for state estimation. 

The problem of generating spoofing attacks has been studied specifically for GPS signals. Tippenhauer et al.~\cite{tippenhauer2011requirements} describe the requirements as well as present a methodology for generating spoofed GPS signals. Larcom and Liu~\cite{larcom2013modeling} presented a taxonomy of GPS spoofing attacks.

The typical approach to mitigate sensor spoofing attacks is by designing robust state estimators~\cite{bezzo2014attack}. Fawzi et al.~\cite{fawzi2011secure} presented the design of a state estimator for a linear dynamical system when some of the sensor measurements are corrupted by an adversarial attacker. We focus on the scenario where the observer uses a Kalman Filter (KF) for estimating the state using measurements that are corrupted by additive spoofing signals by the attackers. We study the problem of generating spoofing signals of minimum energy that can achieve any desired separation between the KF estimate with spoofing and without spoofing. We show that for many practical cases, the spoofing signals can be generated using linear programming in polynomial time. 

Many recent works have undertaken research for design spoofing data against a healthy estimation environment. Such as LQG control system~\cite{mo2010control}, GPS system~\cite{su2016stealthy}, wireless sensor networks~\cite{mo2010false} and electric power grids~\cite{liu2011false}.

In~\cite{liu2011false}, the author presents false data injection attacks, against state estimation in electric
power grids. This paper shows that an attacker can exploit the configuration of a power system to launch such attacks to successfully introduce arbitrary errors into certain state variables while bypassing existing techniques for bad measurement detection.

Another work by Su et al.~\cite{su2016stealthy} is closely related to ours. The authors show how to spoof the GPS signal without triggering a detector that uses the residual in the Kalman filter. They present a 1-step (greedy) online spoofing strategy that solves a linear relaxation of a Quadratically Constrained Quadratic Program (QCQP) at each timestep. We present a strategy that plans for $T$ future timesteps, instead of just the next timestep, while minimizing the spoofing signal energy. Furthermore, we characterize the scenarios under which our strategy finds the optimal solution in polynomial time.

The work that is most closely related to ours is by Mo et al.~\cite{mo2010control,mo2010false}. Their goal is to design false measurement data to mislead a system with Kalman filter~\cite{mo2010false} or an LQG control system~\cite{mo2010control}. Both, our work and the aforementioned work, assume that the system is linear with Gaussian noise and that a discrete Kalman filter is used to estimate the state.  However, in~\cite{mo2010false}, the objective is to design the false data to mislead a certain failure detector ($\chi^2$ failure detector). The paper gives an inner and outer approximation for a reachable set that can mislead the system while not being detected by the $\chi^2$ failure detector. 

Various failure detectors have been proposed in the literature. Jones~\cite{jones1973failure} presented one of the first work on failure detection in linear systems. They presented a linear filter that increases the sensitivity of the residual of the filter, which helps to improve the detection of a particular failure. Brumback et al.~\cite{brumback1987chi} presented a $\chi^2$ test for fault detection in Kalman filters. Mo et al.\cite{mo2010false} studied the effect of false data injection attacks on state estimation with a $\chi^2$ failure detectors. 

In this paper, we study how to design spoofing signals that is agnostic to the failure detector. Instead, we minimize the magnitude of the injected signals while still ensuring the desired separation in the filter output. We provide numerical simulations to show our strategy successfully misleads the $\chi^2$ detector. 

Based on the motion model of the target and the evolution of the KF, three problems for spoofing design are formulated in Section \ref{sec:probform}. Section \ref{sec:Strategies} shows the approaches to solve these optimization problems. The simulations for verifying spoofing strategies are given in Section~\ref{sec:simulation}. Section~\ref{sec:detector} provides a numerical example to illustrate how the proposed spoofing strategy can be applied to a system equipped with a failure detector. Finally, Section~\ref{sec:conc} summarizes the conclusion and future work. 

\section{Problem Formulation} \label{sec:probform}
\noindent\textbf{Notation:} We denote the set of positive real number by $\mathbb{R}^{+}$, the set of positive integer by $\mathbb{Z}^{+}$. The set of real vectors with dimension $n$ is denoted by $\mathbb{R}^{n},~n\in \mathbb{Z}^{+}$, and the set of real matrices with $m$ rows and $n$ columns by $\mathbb{R}^{m\times n}, ~m, n\in \mathbb{Z}^{+}$. We write $\lVert \cdot \lVert_{p}^{p}, ~p\in\mathbb{Z}^{+}$ as the $p^{th}$ power of $L_{p}$ vector norm,  $\mathbb{E}(\cdot)$ as the expectation of a random variable, $I_{n}$ as the identity matrix with size $n, ~n\in \mathbb{Z}^{+}$, and $\mathcal{N}(\mu,\sigma^{2})$ as the normal distribution with mean $\mu$  and variance $\sigma^{2}$.

We consider a scenario where an observer estimates the location of a target using a  KF in 2D plane. The target misleads the observer by adding spoofing signals to the observer's measurement. We define the target's  model as: 
\begin{equation}  \label{eqn:target_motion_model}
  x_{t+1}=\mathcal{F}x_{t}+\mathcal{G}u_t+\omega_t,
\end{equation}   
where $\mathcal{F},\mathcal{G}\in\mathbb{R}^{2\times 2}$, $x_{t}\in\mathbb{R}^2$ is the position of the target, $u_t\in \mathbb{R}^{2}$ is the  control input and $w_t\sim \mathcal{N}(0, R)$ is the Gaussian distribution,  model noise of the motion model with $R\in \mathbb{R}^{2\times 2}$. 

The observer estimates the target's position  using a linear measurement model:
\begin{equation}
  \label{eqn:realmeasurement}
   z_t= \mathcal{H}x_t +v_t,
\end{equation}
where  $\mathcal{H}\in \mathbb{R}^{2\times 2}$ and $v_t \sim \mathcal{N}(0,Q)$ gives the measurement noise with $Q\in \mathbb{R}^{2\times 2}$.

In order to mislead the observer, the target corrupts the observer's measurement by adding spoofing signal to mislead the observer's estimate. We assume the measurement received by the observer is $\tilde{z}_t\in \mathbb{R}^2$ with spoofing signal (Equation~\eqref{eqn:spoofmeasurement}) instead of the true measurement $z_t\in \mathbb{R}^2$ without spoofing signal (Equation~\eqref{eqn:realmeasurement}). The spoofing signal $\epsilon_t:=[\epsilon_{tx}, \epsilon_{ty}]^T\in \mathbb{R}^2$ adds additional measurement error: 
\begin{equation}
  \label{eqn:spoofmeasurement}
 \tilde{z}_t= z_t + \epsilon_t.
\end{equation}

\begin{figure}
  \centering
  \includegraphics[width=0.8\columnwidth]{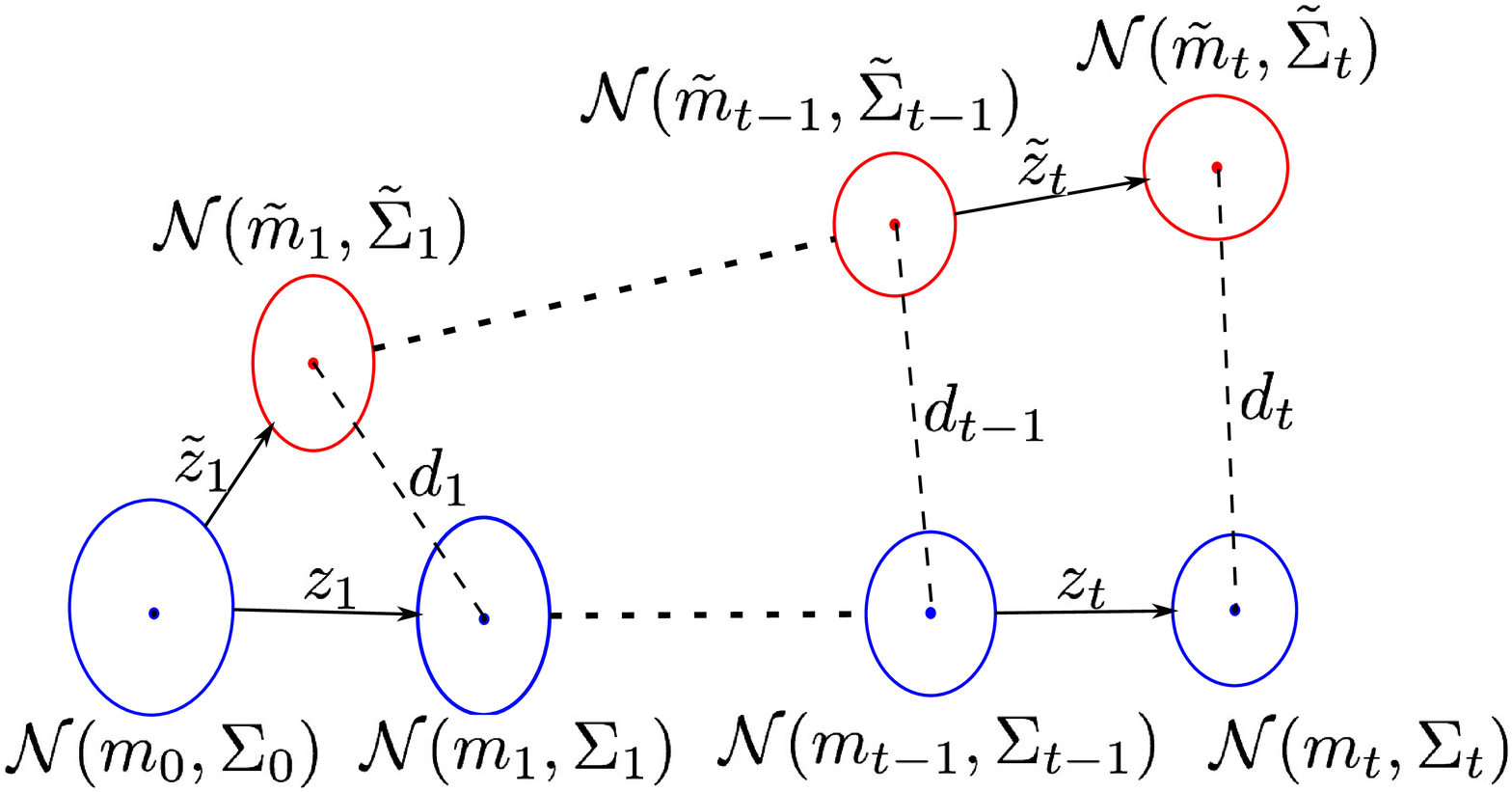}
  \caption{The evolution of KF estimate by applying $z_t$ and  $\tilde{z}_t$, respectively.}
  \label{Distance}
\end{figure}

The observer uses a KF to estimate target's position with initial distribution $\mathcal{N}(m_0, \Sigma_0)$. Since it receives the spoofing measurement $\tilde{z}_t$ for updating, we denote distributions generated by the evolution of its KF as $\mathcal{N}(\tilde{m}_t, \tilde{\Sigma}_t)$ when step $t\geq 1, t\in \mathbb{Z}^{+}$. We also denote the distributions generated by the evolution of a KF using true measurement $z_t$ as $\mathcal{N}({m}_t, {\Sigma}_t)$. The goal for the target is to set the separation between the mean estimate $m_t$ and $\tilde{m}_t$. The target's spoofing signal is each step within the planning horizon for which some desired separation, $d_t\geq 0$, must be achieved (Figure~\ref{Distance}). Figure~\ref{explain_distance} shows the target's spoofing process where it uses the initial guess of $\mathcal{N}(m_0, \Sigma_0)$ denoted as $\mathcal{N}(\tilde{m}_0, \tilde{\Sigma}_0)$ and desired separation $d_t$ to design spoofing signal $\epsilon_t$. In order to avoid detection, the targets seeks to minimize the magnitude of the spoofing signal.
\begin{figure}
  \centering
  \includegraphics[width=0.75\columnwidth]{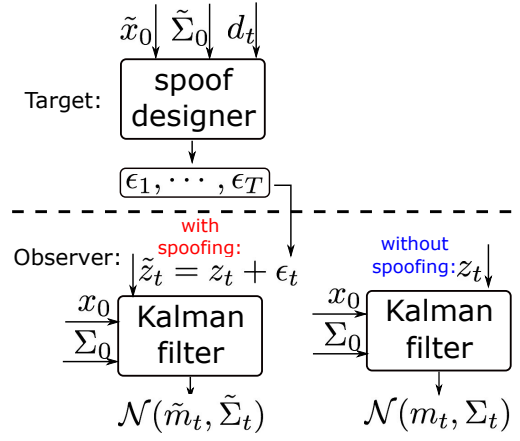}
  \caption{Signal spoofing process and its effect on the observer's KF estimation.}
  \label{explain_distance}
\end{figure}
We first propose two problems for offline scenarios as follows.
\subsection{Offline Spoofing Signal Design with Known  $\mathcal{N}(m_0, \Sigma_0)$}
If the target knows $\mathcal{N}(m_0, \Sigma_0)$ of the KF, then the target can set $\mathcal{N}(\tilde{m}_0, \tilde{\Sigma}_0)$ equal to $\mathcal{N}(m_0, \Sigma_0)$.  

\begin{problem} [Offline with Known $\mathcal{N}(m_0, \Sigma_0)$] \label{pro:problem1}
Consider a target with motion model (Equation~\eqref{eqn:target_motion_model}), measurement model (Equation~\eqref{eqn:realmeasurement}), and spoofing measurement model (Equation~\eqref{eqn:spoofmeasurement}). Assume target knows $\mathcal{N}(m_0, \Sigma_0)$. Find a sequence of spoofing signal inputs, $\{\epsilon_1, \epsilon_2,\cdots,  \epsilon_{T} \}$ to achieve desired separation $d_t$  between $\tilde{m}_t$ and $m_t$ at step $t$. Such that,
\begin{equation*}
  \label{eqn:obj_pro_1}
 \text{minimize}  \quad \sum^{T}_{t=1} \gamma_t \cdot \lVert \epsilon_t\lVert_p^p
\end{equation*}
subject to,
\begin{equation} \label{pro:problem1stconstraint}
\begin{split}
\quad \lVert m_t -& \tilde{m}_t\lVert_{p}^{p} \ge d_{t}^{p},\quad \forall t
      \end{split}
\end{equation}
where $\gamma_t\in\mathbb{R}^{+}$ is a weighing parameter and $T\in \mathbb{Z}^{+}$ is the optimization horizon. 

\end{problem}
\subsection{Offline Spoofing Signal Design with Unknown  $\mathcal{N}(m_0, \Sigma_0)$}
Next, we consider the case where the target does not know the initial condition in the KF. Instead, we assume that the initial estimate $\tilde{m}_0$ is not too far away from $m_0$ (in exception). 
\begin{problem}  [Offline with Unknown $\mathcal{N}(m_0, \Sigma_0)$] \label{pro:problem2}
Consider a target with motion model (Equation~\eqref{eqn:target_motion_model}), measurement model (Equation~\eqref{eqn:realmeasurement}), and spoofing measurement model (Equation~\eqref{eqn:spoofmeasurement}). Assume the target starts spoofing with $\tilde{m}_0$, where $ \mathbb{E}( m_0 - \tilde{m}_0) = M_0$  and $\tilde{\Sigma}_0 \neq {\Sigma}_0$. Find a sequence of spoofing signal inputs, $\{\epsilon_1, \epsilon_2,\cdots,  \epsilon_{T} \}$ to achieve desired separation $d_t$  between $\tilde{m}_t$ and $m_t$ (in expectation) at step $t$. Such that
\begin{equation*}
  \label{eq5}
 \text{minimize}   \quad \sum^{T}_{t=1} \gamma_t \cdot \lVert \epsilon_t\lVert_p^p 
\end{equation*}
subject to,
\begin{equation} \label{problem2st} 
\begin{split}
\quad \lVert \mathbb{E}(m_t -& \tilde{m}_t)\lVert_{p}^{p} \ge d_{t}^{p}, \quad \forall t
\end{split}
\end{equation}
where $\gamma_t\in\mathbb{R}^{+}$ is a weighing parameters and $T\in\mathbb{Z}^{+}$ is the optimization horizon.
\end{problem}
\section{Signal Spoofing Strategies} \label{sec:Strategies}
In this section, we show how to solve Problems~\ref{pro:problem1} and~\ref{pro:problem2} when $p=1$ and $p=2$. We first present the relationship between the separation $m_t - \tilde{m}_t$  and the initial bias $m_0 - \tilde{m}_0$.
\begin{theorem} \label{mk} 
Consider a target with motion model (Equation~\eqref{eqn:target_motion_model}), measurement model (Equation~\eqref{eqn:realmeasurement}), and spoofing measurement model  (Equation~\eqref{eqn:spoofmeasurement}). The evolutions of the KFs by applying $z_t$ and $\tilde{z}_t$ give the distributions $\mathcal{N}(m_t, \Sigma_t)$ and $\mathcal{N}(\tilde{m}_t, \tilde{\Sigma_t})$, respectively. The difference, $m_t  - \tilde{m}_t$ is,
\begin{equation} \label{mt}
\begin{split}
m_t  - \tilde{m}_t = &\prod_{i=0}^{t-1}A_{t-i}\cdot (m_0  - \tilde{m}_0) + \\
&\sum_{i=0}^{t-2} \left( \prod_{j=0}^{i}A_{t-j}  (B_{t-1-i} + C_{t-1-i}) \right) +B_t + C_t,
\end{split}
\end{equation}
where $A_t = \mathcal{F} - \tilde{K}_t \mathcal{H}\mathcal{F},\quad B_t = (K_t-\tilde{K}_t)\left[z_t-\mathcal{H}(\mathcal{F}m_{t-1}+\mathcal{G}u_{t-1})\right], \quad
C_t= - \tilde{K}_t \epsilon_t$. 
\end{theorem}

The proof is given in the appendix. 

\begin{corollary} \label{corr:expected}
The expected value of the separation is, 
\begin{equation}
\begin{split} \label{mt_expectation}
&  \mathbb{E}\left(m_t-\tilde{m}_t\right)  \\
&=  \prod_{i=0}^{t-1}A_{t-i} M_0 +\sum_{i=0}^{t-2} \left(\prod_{j=0}^{i}A_{t-j} C_{t-1-i}\right) +  C_t.
\end{split}
\end{equation}
\end{corollary}
\begin{proof}
From Equation~\ref{mt}, $\mathbb{E}\left(m_t-\tilde{m}_t\right)$ follows,
\begin{equation*} 
\begin{split}
&  \mathbb{E}\left(m_t-\tilde{m}_t\right)   \\
&=  \mathbb{E} \left(\sum_{i=0}^{t-2}  \prod_{j=0}^{i}A_{t-j}\cdot B_{t-1-i}    +  B_t \right) + \\
&\quad\prod_{i=0}^{t-1}A_{t-i} \mathbb{E} (m_0  - \tilde{m}_0) +\sum_{i=0}^{t-2} \left( \prod_{j=0}^{i}A_{t-j} C_{t-1-i}\right) +  \tilde{K}_t\epsilon_t.
\end{split}
\end{equation*}
The actual measurement is: $z_i = \mathcal{H}(\mathcal{F}m_{i-1} +\mathcal{G}u_{i-1}+w_i)+v_i$, where $w_i$ and $v_i$ are Gaussian noises with zero mean. The expected measurement value is: $\mathbb{E}(z_i)= \mathcal{H}( \mathcal{F}m_{i-1}+ \mathcal{G}u_{i-1})$ for all $i$, thus $\mathbb{E}[z_i- \mathcal{H}( \mathcal{F}m_{i-1}+ \mathcal{G}u_{i-1})]=0$.  Since $\mathbb{E}[B_i] = 0$, we have, 
\begin{equation}
\begin{split} 
&  \mathbb{E}\left(m_t-\tilde{m}_t\right)  \\
&=  \prod_{i=0}^{t-1}A_{t-i}\mathbb{E} (m_0  - \tilde{m}_0) +\sum_{i=0}^{t-2} \left(\prod_{j=0}^{i}A_{t-j} \tilde{K}_{t-1-i} \epsilon_{t-1-i}\right)\\
&\quad +  \tilde{K}_t\epsilon_t.
\end{split}
\end{equation}
Since we assume $ \mathbb{E}( m_0 - \tilde{m}_0) = M_0$ in Problem~\ref{pro:problem2}, the claim is guaranteed.
\end{proof}


Theorem \ref{mk} shows the difference between the two estimated means at step $t$ depends on the initial means, $m_0$ and $\tilde{m}_0$, and the initial covariance matrices $\Sigma_0$ and $\tilde{\Sigma}_0$. This is because the Kalman gain $K_t$ depends on the covariance matrix $\Sigma_t$. If target sets $m_0=\tilde{m}_0$ and $\Sigma_0 = \tilde{\Sigma}_0$, it has $\Sigma_t = \tilde{\Sigma}_t$ for all $t$ since the covariance matrix is updated through the same Kalman prediction and update equation (see appendix). Thus, $B_t = 0_{2\times 2}$ and then Equation \eqref{mt} can be simplified as:
$$  m_t  - \tilde{m}_t =\sum_{i=0}^{t-2} \left(\prod_{j=0}^{i}A_{t-j}  C_{t-1-i}\right) + C_t. $$
As a result, $m_t  - \tilde{m}_t $ is independent of the measurements $\{z_1, z_2, \cdots, z_t\}$ when $m_0=\tilde{m}_0$ and $\Sigma_0 = \tilde{\Sigma}_0$. Thus, the target can generate spoofing signal inputs by solving Problem \ref{pro:problem1} offline. Similarly, following Corollary~\ref{corr:expected}, Problem~\ref{pro:problem2} can be saved offline as well.

 Problems~\ref{pro:problem1} and~\ref{pro:problem2} are two nonlinear programming problems for arbitrary vector norms $L_p$. However, 
when $p=1$, they can be formulated as linear programming problems. Linear programming can be solved in polynomial time \cite{karmarkar1984new}. When $p=2$, they become QCQP (Quadratically Constrained Quadratic Program). The following shows the LP and QCQP formulations.
\begin{theorem}
If $p=1$ and the elements in $ \mathcal{F}$ and $I-K_t \mathcal{H}$ are all positive, then Problems~\ref{pro:problem1} and \ref{pro:problem2} can be solved optimally with linear programming. If $p=1$ and the elements in $ \mathcal{F}$ and $I-K_t \mathcal{H}$ are not all positive, then Problems~\ref{pro:problem1} and \ref{pro:problem2} can be solved optimally with $4^k$ linear programming instances. If $p=2$ and $\{ \mathcal{H}, \mathcal{F},Q,R\}$ are diagonal matrices, then Problems~\ref{pro:problem1} and \ref{pro:problem2} can be solved optimally with linear programming. 
\end{theorem}

\subsection{Linear Programming Formulation  for $L_1$ Vector Norm}
\label{sec:$L_1$ norm}
Here, we show how to formulate Problem~\ref{pro:problem1} using linear programming. A similar procedure can be applied to formulate Problem~\ref{pro:problem2} as linear programming.

The constraint in Problem~\ref{pro:problem1} (Equation~\ref{pro:problem1stconstraint}) follows:
\begin{align}\label{replacePro1}
&\lVert m_t-\tilde{m}_t\lVert_1  =\left\lVert\sum_{i=0}^{t-2} \left(\prod_{j=0}^{i}A_{t-j}  C_{t-1-i}\right) + C_t \right\lVert_1\nonumber\\
&= \left\lVert \sum_{i=0}^{t-2} \left( \prod_{j=0}^{i}A_{j+1}\cdot \tilde{K}_{t-1-i}  \cdot \epsilon_{t-1-i}\right) + \tilde{K}_t \epsilon_t \right\lVert_1\nonumber\\
&\geq d_t,
\end{align}
where $t=1,2,\cdots,T$. $\prod_{j=0}^{i}A_{t-j}\cdot \tilde{K}_i\in \mathbb{R}^{2\times 2}$ is a constant matrix for each $i\in\{1,\cdots,t-1\}$ and is calculated from the KF iteration with initial covariance $\Sigma_0$ and $\tilde{\Sigma}_0$. 
Since $L_{1}$ vector norm is the sum of the absolute values of the elements for a given vector, Problem~\ref{pro:problem1} can be directly formulated as a linear programming problem when $p=1$. 

Then we show how to transform this constraint to a standard linear constraint form $G_t x_t \ge d_t$ with $x_t:=[\epsilon_{1x},\cdots,\epsilon_{tx}, \epsilon_{1y},\cdots,\epsilon_{ty}]^{T}$. The left side of Equation \eqref{replacePro1} can be formulated as
\begin{equation}
\begin{split}
\lVert m_t - \tilde{m}_t\lVert_1=
\bigg\lVert                 
  \begin{array}{cc} 
   a_0+ a_{1} \epsilon_{1x} +\cdots a_{t}\epsilon_{tx} +\cdots+ a_{2t}\epsilon_{ty}  \\
   b_0+ b_{1} \epsilon_{1x} +\cdots b_{t}\epsilon_{tx} + \cdots + b_{2t}\epsilon_{ty} \\
  \end{array}
\bigg\lVert_1 
\label{eqn:specific_constraint_example}
\end{split}
\end{equation}
where $a_0,a_{1},\cdots,a_{2t},b_0,b_{1},\cdots,b_{2t}$ are corresponding coefficients from Equation~\ref{mt}.
\begin{lemma}
\label{lemma_linear}
If the elements in matrices $ \mathcal{F}$ and $I - K_t \mathcal{H}$ are positive, then  $\lVert m_t - \tilde{m}_t\lVert_1$ is a linear combination of $|\epsilon_{ix} |$ and $|\epsilon_{iy}|$, and Problem~\ref{pro:problem1} can be solved as a single LP instance.
\end{lemma}
\begin{proof}
According to the proof of Theorem~\ref{mk} appendix, all the coefficients $\{a_1,...,a_{2t}, b_1,...,b_{2t}\}$ are positive if the elements in matrices $ \mathcal{F}$ and $I - K_t \mathcal{H}$ are positive. Therefore, the objective function and the constraints are linear in $|\epsilon_{ix} |$ and $|\epsilon_{iy}|$. There always exists an optimal solution where all $\epsilon_{ix} \geq 0$ and $\epsilon_{iy}\geq 0$ or where all $\epsilon_{ix} \leq 0$ and $\epsilon_{iy}\leq 0$. The objective function in both cases will be the same. Without loss of generality, we can assume $\epsilon_{ix} \geq 0$ and $\epsilon_{iy}\geq 0$, which can be solved using a single LP instance.
\end{proof}

The linear programming strategy  containing $k$ constraints is presented in Algorithm~\ref{alg:linear_programming}. $G$ denotes matrix in the linear constraint $G x \geq D_k$ where $x:=[\epsilon_{1x},\cdots,\epsilon_{Tx}, \epsilon_{1y},\cdots,\epsilon_{Ty}]^{T}$ and $D_k$ is the collection of $k$ nonzero separations $d_t, ~t\in\{1,\cdots, T\}$.
\begin{algorithm}\label{alg:linear_programming}
    \SetAlgoLined
$\mathbf{Initial}\leftarrow\left\{(x_o, \Sigma_0, \mathcal{F},\mathcal{H},  \mathcal{G}, Q, R, u \right\}$\\  
$G\leftarrow 0_{k\times 2T}$ \\
Calculate Kalman gain $\tilde{K}_1,\cdots,\tilde{K}_T$ \\
\For{$i=1$ to the $q_{th}$ value in $D_k$}{
$g = \prod_{j=i}^{T-1}A_{j+1} \tilde{K}_i$\label{algline:g}\;
$G_{q,i} = \text{sum  of  all  rows in}$ $g$
}
$\mathbf{Return} \quad G$
    \caption{Linear Programming Formulation}
    \label{minimax} 
\end{algorithm}

If Lemma \ref{lemma_linear} does not  hold, it is possible that some elements in $a_0,a_{1},\cdots,a_{2t},b_0,b_{1},\cdots,b_{2t}$ can be positive and some are negative. In general,  there are four different cases depending on the sign of the first row and the second row for considering each constraint $\lVert m_t - \tilde{m}_t\lVert_1 \geq d_t$ (Equation~\ref{eqn:specific_constraint_example}). Then we can obtain four linear optimization problems along four different sub-constraints of each constraint $\lVert m_t - \tilde{m}_t\lVert_1 \geq d_t$. Thus, in the worst case, the optimal solution can be obtained by solving $4^k$ linear optimization problems. We run Algorithm~\ref{alg:linear_programming} $4^k$ times by changing the sign of rows in $g$ (Line~\ref{algline:g}) appropriately.

\subsection{Quadratically Constrained Quadratic Program Formulation for $L_2$ Vector Norm}
When $p=2$, Problems~\ref{pro:problem1} and~\ref{pro:problem2} can be formulated as  QCQPs\cite{boyd2004convex}:

\begin{equation}
\begin{split}
&\text{minimize} \quad  \frac{1}{2}x_{\epsilon}^{T}P_0 x_\epsilon \\
&\text{s.t.}\quad -\frac{1}{2} x^T_\epsilon D^T_t D_t x_\epsilon + d_{t}^{2} \le 0 ,\quad\forall t\in\{1,\dots,T\}\\
\label{eqn:QCQP_forthree}
\end{split}
\end{equation}
where $x_{\epsilon}=[\epsilon^2_{1x},\epsilon^2_{1y},\cdots,\epsilon^2_{Tx},\epsilon^2_{Ty}]^{T}$, $P_0=I_{2T}$, and 

$D_t\in \mathbb{R}^{2T \times 2T} :=$ 
\[
  \begin{bmatrix}
    \prod_{j=1}^{t-1}A_{j+1} \tilde{K}_0 & &\cdots &0 & 0 &0\\
    \vdots& \ddots &\vdots &\vdots &\vdots &\vdots\\
    0& \cdots& \prod_{j=t-1}^{t-1}A_{j+1} \tilde{K}_{t-1} &0 &0 &0\\
    0& \cdots& 0&  \tilde{K}_t & 0 & 0\\
    0& \cdots& 0&  0 & 0 & 0\\
    0& \cdots& 0&  0 & 0 & \ddots
  \end{bmatrix}
\]

Unfortunately, the QCQP formulations for these three problems are NP-hard since the constraint in each problem is concave. If $ \mathcal{F}, \mathcal{G}, \mathcal{H},\tilde{\Sigma_0}$ are diagonal matrices, it can be shown that $D_t$ is also a diagonal matrix. We can transform the QCQP formulation to a linear programming problem by using change of variables $\{\epsilon^2_{tx},\epsilon^2_{ty}\}, ~t=\{1,2,...,T\}$, and using  a procedure similar to $p=1$.

If $D_t$ is not a diagonal matrix, one solution is to apply the inequality $\sqrt[]{2} \lVert x \lVert_2 \ge  \lVert x \lVert_1$ between $L_1$ vector norm and $L_2$ vector norm. The constraint can be changed to  $L_1$ vector norm, which is a stricter constraint. A  sub-optimal solution can be obtained by using the $L_1$ vector norm. 
\subsection{Receding Horizon: Spoofing with online measurement}\label{pro:problem3}
 Problems \ref{pro:problem1} and \ref{pro:problem2} describe the offline versions for spoofing. We also extend the offline problems to an online version. The following formulates an online spoofing  scenario.
 
Consider a target with motion model (Equation~\eqref{eqn:target_motion_model}), measurement model (Equation~\eqref{eqn:realmeasurement}), and spoofing measurement model (Equation~\eqref{eqn:spoofmeasurement}). Assume the target does not know $\mathcal{N}({x}_0, {\Sigma}_0)$. It collects a series of measurements $\{z^{real}_1, z^{real}_2, \cdots, z^{real}_{t^o} \}$ from step $1$ to current step $t^{o}$. Find a sequence of spoofing signal inputs, $\{\epsilon_{t^{o}}, \epsilon_{t^{o}+1},\cdots,  \epsilon_{t^{o}+H} \}$ to achieve desired separation $d_t$  between $\tilde{m}_t$ and $m_t$ (in expectation) within future $H$ steps. Such that
\begin{equation*}
  \label{eq5}
  \text{minimize}   \quad \sum^{t^o+H}_{t=t^{0}} \gamma_t \cdot \lVert \epsilon_t\lVert_p^p
\end{equation*}
\begin{equation} \label{pro:pro3_constrain}
\begin{split}
\text{s.t.} \quad \lVert \mathbb{E}(m_t -& \tilde{m}_t)\lVert_{p}^{p} \ge d_{t}^{p} ,\quad \forall t\in\{t^o,\cdots,t^o+H\}
\end{split}
\end{equation}
where $\gamma_t\in\mathbb{R}^{+}$ is a weighing parameter, $t^o$ is the current time, and $H$ is the predictive time horizon. The target applies $\epsilon_{t}=\epsilon_{t^o}$ as spoofing signal input at each step $t$.

\section{Simulations} \label{sec:simulation}
In this section, we simulate the effectiveness of spoofing strategies for Problems~\ref{pro:problem1},~\ref{pro:problem2} and online case (Section~\ref{pro:problem3}) where a target designs spoofing signals $\epsilon_t$ to  mislead an observer by achieving the desired separations $d_t$ between $m_t$ and $\tilde{m}_t$. Our code  is available online.\footnote{\url{https://github.com/raaslab/signal_spoofing.git}}

We consider the $L_1$  vector norm and the following models,
$$ \mathcal{F}=I_{2\times2},  \mathcal{G}=I_{2\times2}, u= \left[\begin{matrix}
  1  \\
   1  \\
  \end{matrix} \right], 
  R =0.5I_{2\times2}, 
    Q =0.5I_{2\times2}.\\$$
Set the weight $\gamma_t=1$ for all $t$. 
 
For Problem~\ref{pro:problem1}, set the initial condition for the KF as,  
$$\Sigma_0 =I_{2\times2},
     ~m_0=\left[\begin{matrix}
  0  \quad  0  
    \end{matrix}\right]^T.
     $$
Since the target knows $\mathcal{N}({x}_0, {\Sigma}_0)$, it sets $\tilde{m}_0=m_0$ and $\tilde{\Sigma}_0 = \Sigma_0$. We first consider a scenario where the target wants to achieve the desired separation at steps, $t=5, 10, 15$, denoted as $d_5=1.77$, $d_{10}=3.54$ and $d_{15}=5.30$ with the optimization horizon $T=20$. The target generates  a sequence of spoofing signals $\{\epsilon_1,\cdots,\epsilon_{20}\}$ offline by using a linear programming solver. The spoofing performance is shown in Figure~\ref{Simulation_problem1}-(a) where the true separations are the same as the desired separations. Same successful spoofing achieved when the desired separations are chosen as $d_t = 0.25 \sqrt{2} t,~t=\{3,...,15\}$, as shown in Figure~\ref{Simulation_problem1}-(b). 

\begin{figure} 
\centering{
\subfigure[Desired separations, $d_5 = 1.77$ and $d_{10}=3.54$, with $T=20$.]{\includegraphics[width=0.8\columnwidth]{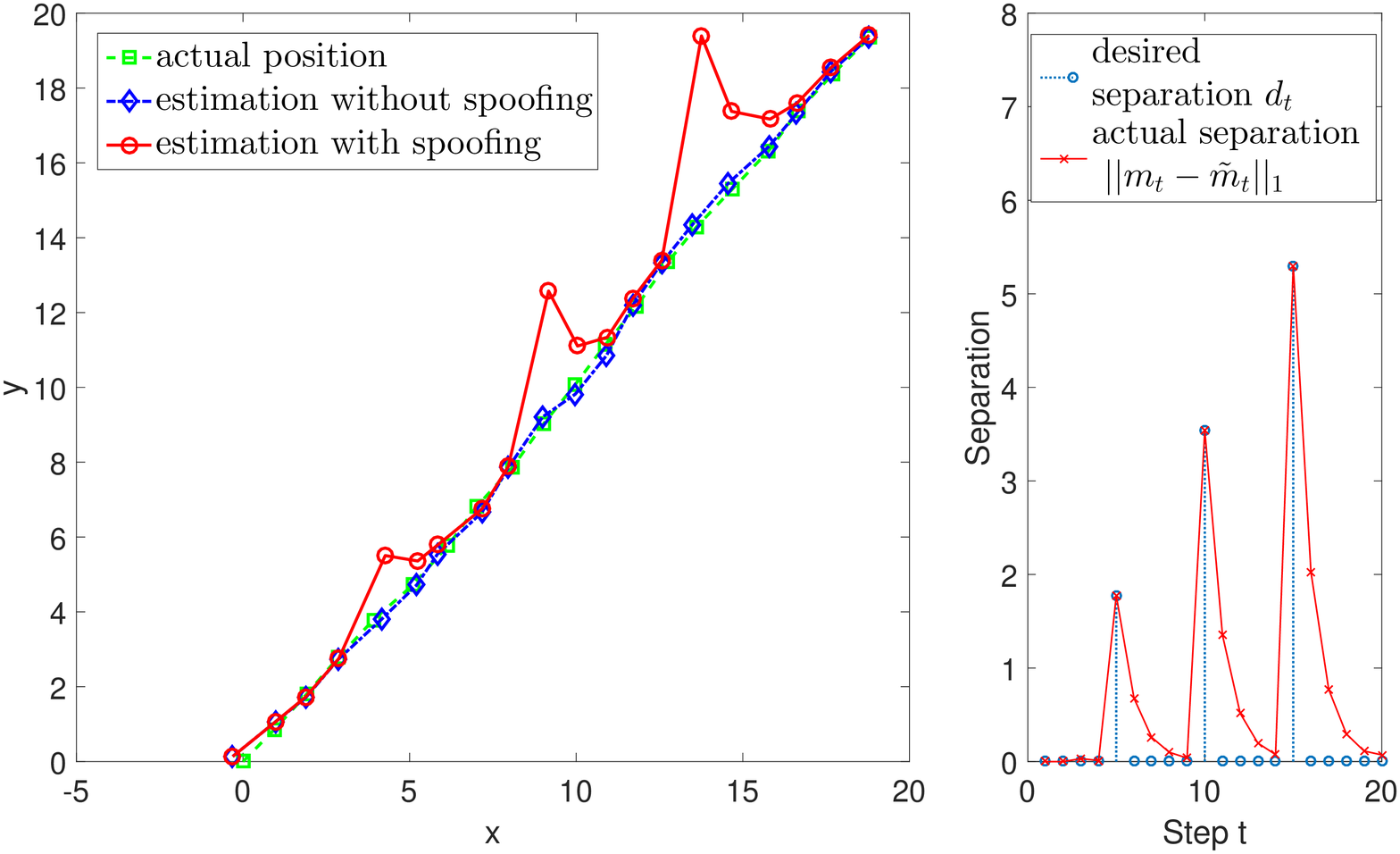}}
\subfigure[Desired separations, $d_t = 0.25 \sqrt{2}t$, with $t=3$ to $T=15$.]{\includegraphics[width=0.8\columnwidth]{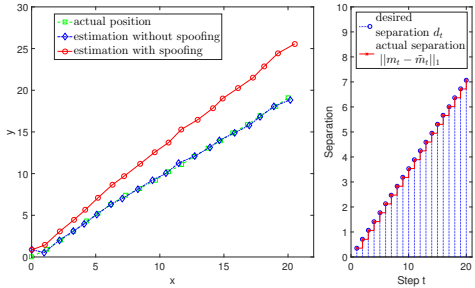}}
}
\caption{Offline signal spoofing with known ($m_0, \Sigma_0$).}
\label{Simulation_problem1}         
\end{figure}
In Problem \ref{pro:problem2}, the target knows $ \mathbb{E}(m_0 - \tilde{m}_0) = M_0$ but does not know $\Sigma_0$. The spoofing result is no longer deterministic but holds in expectation $\lVert\mathbb{E}( m_t - \tilde{m}_t)\lVert_1 \ge d_t$. Figure~\ref{Simulation_problem2}-(a) shows spoofing signals for desired separations as $d_1=2$ with  $T=6$ and $M_0=1$. Set $\mathcal{N}(\tilde{m}_0,\tilde{\Sigma}_0)$ as $\mathcal{N}(0,1.5I_{2})$, $m_0$ as a random variable ($m_0 \sim \mathcal{N}(1,1)$) and ${\Sigma}_0 = I_{2}$. In order to see the effectives of the spoofing signals $\{\epsilon_1,\cdots,\epsilon_5\}$, we conduct 100 trials for each desired separation $d_2\in\{1,2,3,4,5\}$. Figure~\ref{Simulation_problem2}-(b) shows the $\lVert m_1 - \tilde{m}_1 \lVert_1$ is no longer deterministic, but  $\lVert \mathbb{E}(m_1 - \tilde{m}_1) \lVert_1 $ is  close to the desired value $d_1=2$. 
\begin{figure} 
\centering
\subfigure[Desired separation, $d_1 = 2$.]{\includegraphics[width=0.8\columnwidth]{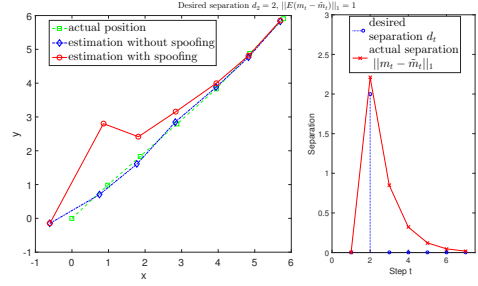}}
\subfigure[Results with $d_1$ = $\{1,2,3,4,5\}$ for 100 trials.]{\includegraphics[width=0.7\columnwidth]{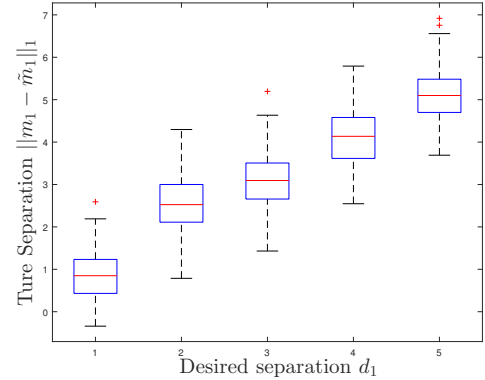}}
\caption{Offline signal spoofing with unknown ($m_0, \Sigma_0$).}
\label{Simulation_problem2}             
\end{figure}

For online case, spoofing signals are continuously generated by using receding horizon optimization with new noisy measurements. We set the receding horizon as $H=15$. Even though offline strategy performs comparatively as online strategy (Figure~\ref{Simulation_online}), online spoofing strategy achieves almost the same separation as the desired, while offline strategy has certain divergence (Figure~\ref{separation_online}). This is because online strategy can update the measurement at each step. Figure~\ref{online_sum} shows the  online strategy applies less total spoofing magnitude than offline strategy.
\begin{figure} 
\centering
{\includegraphics[width=0.7\columnwidth]{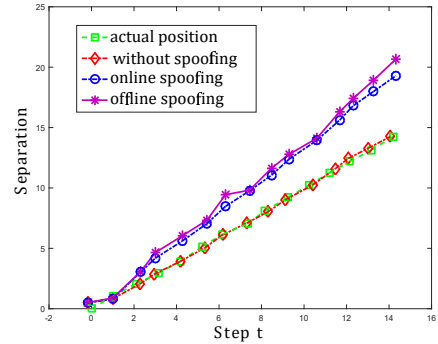}}
\caption{Online spoofing and offline spoofing with unknown ($m_0, \Sigma_0$).}
\label{Simulation_online}           
\end{figure}
\begin{figure}
\centering
{\includegraphics[width=0.7\columnwidth]{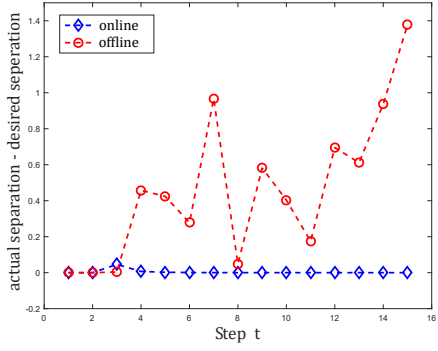}}
\caption{Divergence caused by online spoofing and offline spoofing. The blue line denotes $(\lVert \tilde{m}_t - m_t \lVert_1 - d_t)$ for online spoofing, and red line denotes $(\lVert \tilde{m}_t - m_t \lVert_1 - d_t)$ for offline spoofing.} 
\label{separation_online}
\end{figure}
\begin{figure}
\centering
{\includegraphics[width=0.7\columnwidth]{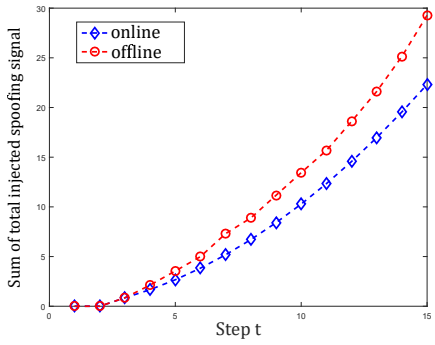}}
\caption{Total spoofing energy injected by online and offline strategies.}
\label{online_sum}
\end{figure}

\section{Signal spoofing with $\chi^2$ failure detector} \label{sec:detector}
In this section, we evaluate the performance of the false data injection strategy in the presence of a failure detector.

We assume a $\chi^2$ failure detector is used in our spoofing design system. A $\chi^2$ detector computes the following measure,
\begin{equation}
g_k = (z_k - \mathcal{H}m_k)^T\Sigma_k (z_k - \mathcal{H}m_k)
\end{equation}
where, $\Sigma_k$ is the covariance matrix, and $z_k - \mathcal{H}m_k$ is the residue of the Kalman filter~\cite{thrun2005probabilistic}. If $g_k>\text{threshold}$, the detector raises an alarm that the filter is under attack.

In general, if the threshold is too low, it may lead to false alarms. On the other hand, a higher threshold is not sensitive enough to detect a true spoofing attack. Therefore, to pick an appropriate value for the threshold, we measured the false alarm rate for a Kalman filter that is not under attack for various values of the threshold. Figure~\ref{False_ararm_rate} shows the false alarm rate for a Kalman filter not under attack. The false alarm rate decreases as the threshold increases. Lienhart et al.~\cite{lienhart2002extended} presented a false alarm detection method which achieves a 12.5\% false alarm rate. The nearest value in our case is a false alarm rate of $11.054\%$ which corresponds to the threshold value of $1.5$. Therefore, for the rest of the evaluation, we use the threshold to be equal to $1.5$.

\begin{figure}[H]
\centering
\includegraphics[width=0.64\columnwidth]{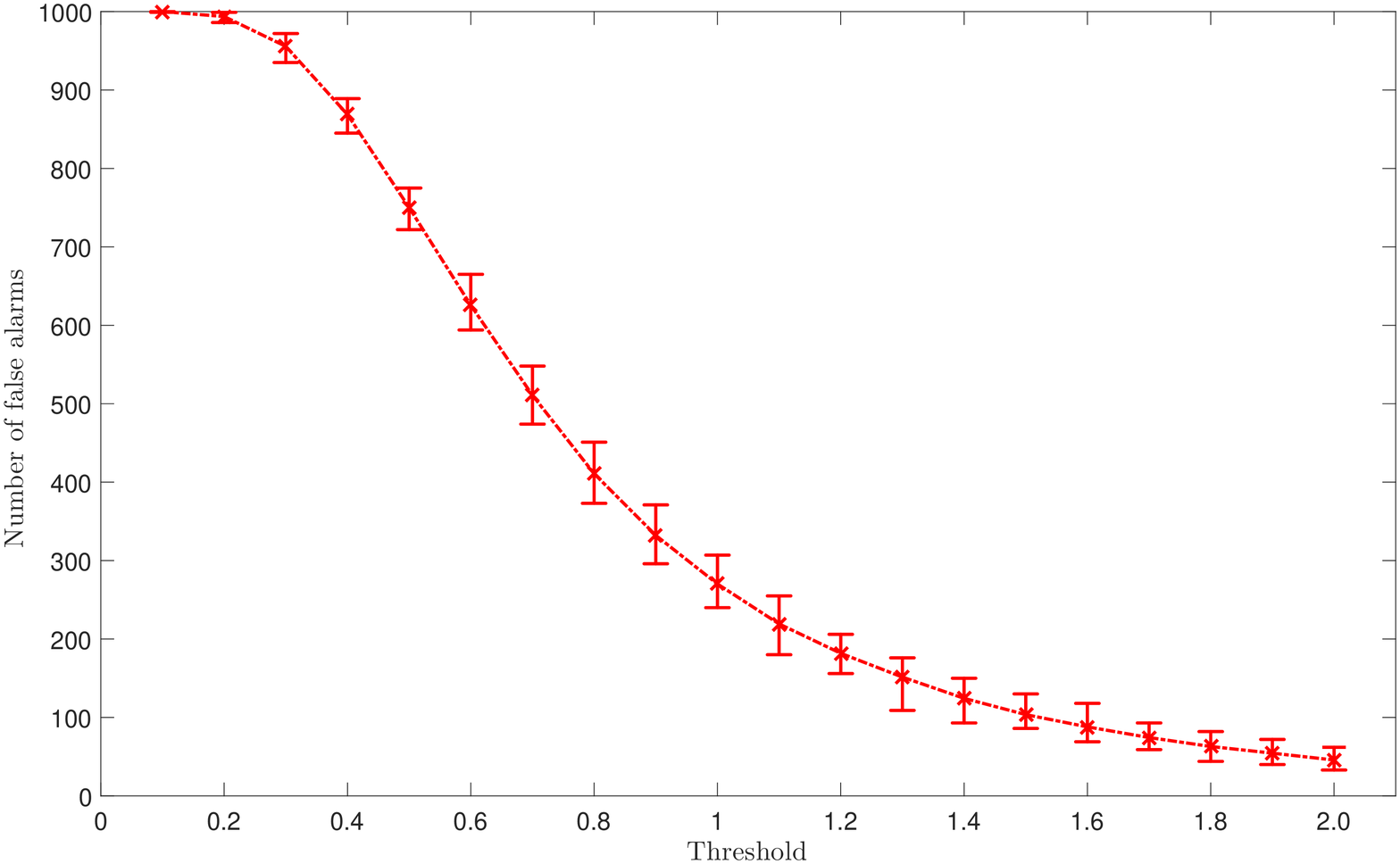}
\caption{False alarm rate. We did 100 simulations for a given value of the threshold. Each simulation consisted of 1000 trials. We plot the mean (red stars) as well as the minimium and maximum number of times the detector raised an alarm for the 100 simulations.}
\label{False_ararm_rate}         
\end{figure}

In the rest of this section, we consider a scenario where the target wants to achieve the desired separation of $d\ge 3$. We consider two strategies to achieve this separation. First, we consider a scenario where the attacker injects spoofing signals abruptly at a specific time instance. Then, we consider a scenario where the attacker gradually injects the spoofing signals to eventually achieve the desired separation. We evaluate the performance in the presence of the $\chi^2$ detector. 

\subsection{Abrupt Failure} 
A simple spoofing strategy is to pick a single time instance to inject the spoofing noise. Here, we choose to inject a spoofing signal to achieve $d_5 = 3$. We ran 1000 trials using our spoofing strategy. As expected, in all 1000 trials, the $\chi^2$ detector was able to detect the attack. One such trial is shown in Figure~\ref{Case1_detectof}. Here, the $\chi^2$ detector output $g_5$ is more than 10 times the threshold and is easily detected. This confirms the intuition that an abrupt failure is likely to be detected.

\begin{figure}[H]
\centering
\includegraphics[width=0.8\columnwidth]{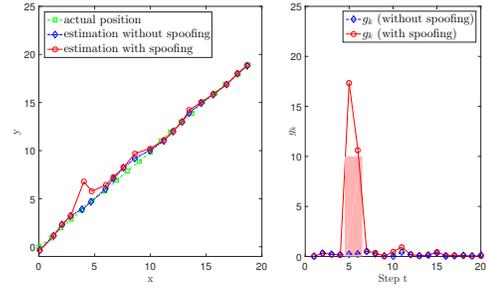}
\caption{Estimated positions and the $\chi^2$ detector's output, $g_k$, when the spoofing signals are injected abruptly by setting $d_5 = 3$. The shaded region indicates when the detector raises an alarm.}
\label{Case1_detectof}         
\end{figure}

\subsection{Gradually Increasing the Separation}
Instead of abruptly injecting spoofing signals, a better strategy is to gradually increase the separation. Here, we choose to gradually inject the spoofing signals over 15 steps. The desired separation is set to $d_t = 0.2 t$. At step $t=15$, we achieve $d_{15} = 3$.

\begin{figure}[H]
\centering
\includegraphics[width=0.8\columnwidth]{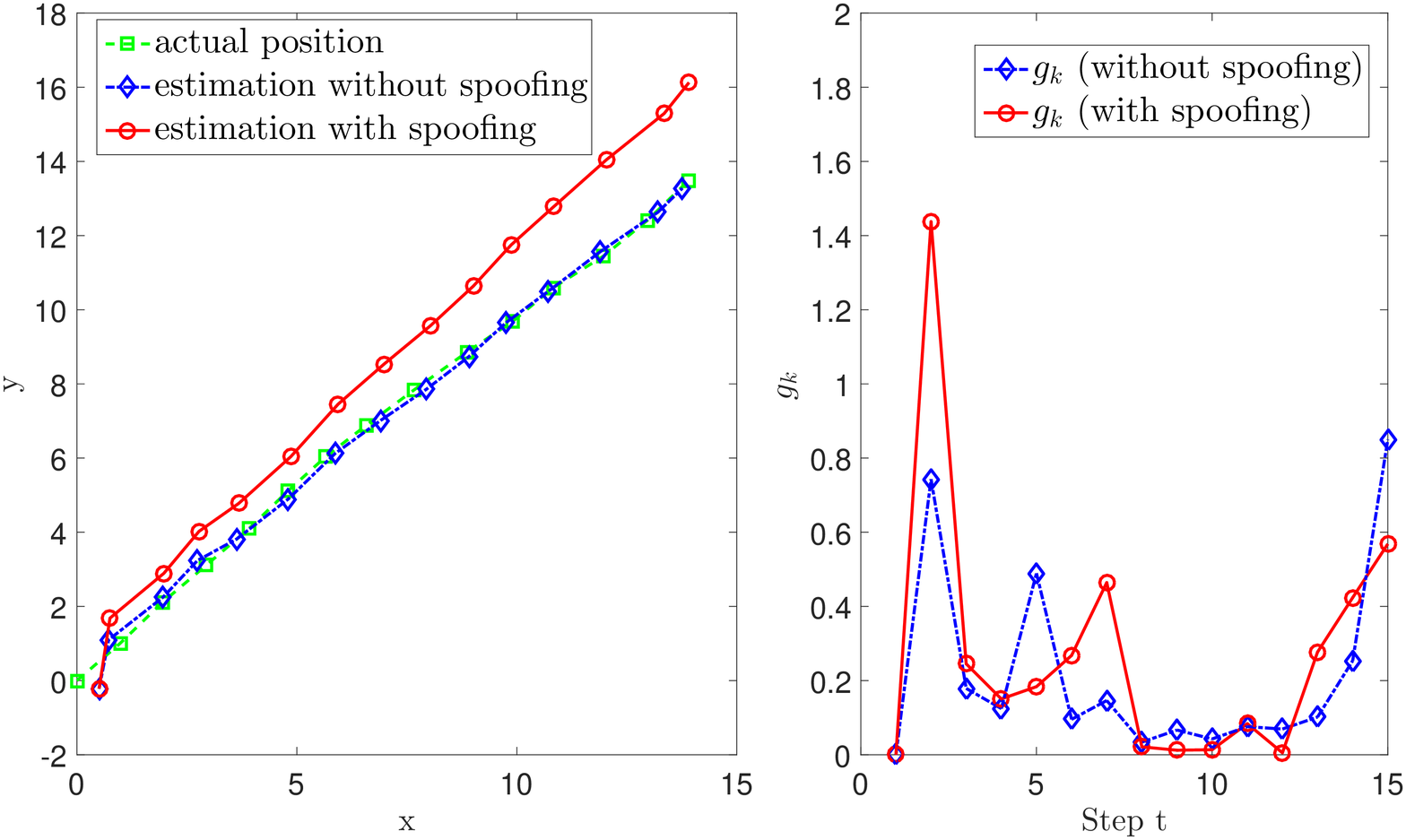}
\caption{Estimated positions and the $\chi^2$ detector's output, $g_k$, when the spoofing signals are injected abruptly by setting $d_t = 0.2 t$.}
\label{Case2_gradually15}         
\end{figure}

We ran 1000 trials using this strategy. Figure~\ref{Case2_gradually15} shows the result of one trial. The $\chi^2$ detector detected the attack in 267 trials. This corresponds to a $p$--value of $9.2133 \times 10^{-43}$. The $p$--value~\cite{chaubey1993resampling} is the probability of detecting an attack in 267 or more trials under the null hypothesis that the system is not under attack (when the false alarm rate is equal to $11.054\%$. Lower $p$--value implies an unlikely event suggesting that the null hypothesis is wrong. A $p$--value of $9.2133 \times 10^{-43}$ is low enough to reject the null hypothesis.

Next, we make the false data injection even more gradual. The desire separation is set as $d_t = 0.1t$. The final separation goal is still $d>3$ which is now achieved at $t=30$ instead of $t=15$.

\begin{figure}[H]
\centering
\includegraphics[width=0.8\columnwidth]{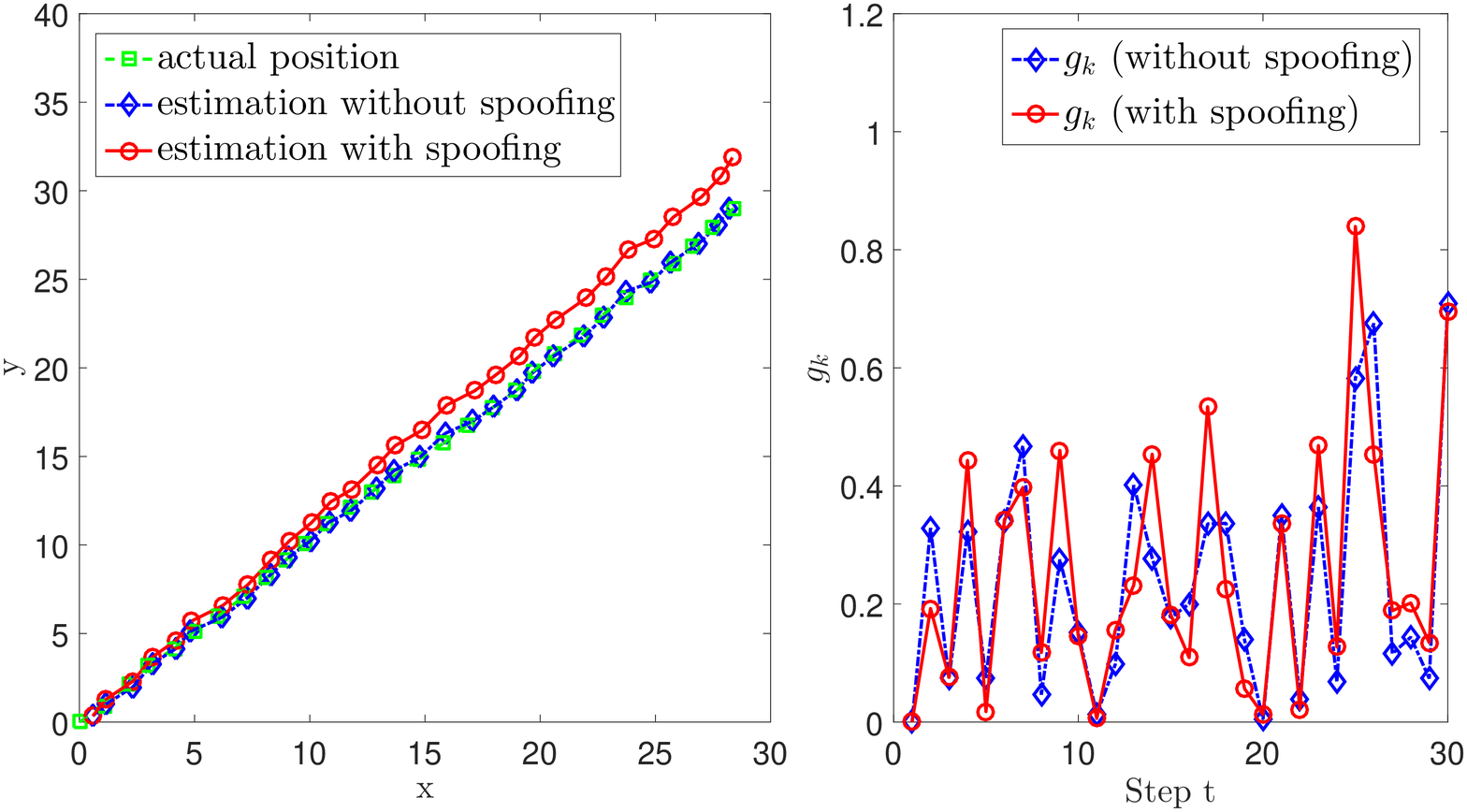}
\caption{Estimated positions and the $\chi^2$ detector's output, $g_k$, when the spoofing signals are injected abruptly by setting $d_t = 0.1 t$.}
\label{Case3_gradually30}         
\end{figure}

In 1000 trials, the $\chi^2$ detector detected an attack in 118 instances. This is close to the actual false alarm rate without any attack. In this scenario, the system will not be able to distinguish between false and true alarms. This is reflected in the $p$--value computation. 118 detections in 1000 trials correspond to a $p$--value of 0.2393. This $p$--value is large enough which incorrectly suggests that the null hypothesis (the system is not under attack) is true. Therefore, in this scenario, the signal spoofing strategy is able to successfully mislead the $\chi^2$ detector. 

As a result, we can see that the $\chi^2$ detector can detect failures in many cases when the signals are injected abruptly. However, using our spoofing strategy, the separation can be carefully and gradually designed so as to make it unlikely for the $\chi^2$ detector to distinguish between false alarms and true alarms. As long as the desired separation is made gradual, the probability of being detected will keep decreasing and eventually make it indistinguishable from false alarms.

\section{Conclusion} \label{sec:conc}

We study the problem of injecting spoofing signals to achieve a desired separation in the output of a Kalman filter without and with attack. We study many variants of the problem. Our main approach was to formulate the problems as nonlinear, constrained optimization problems in order to minimize the energy of the spoofing signal. We show that under some technical assumptions, the problems can be solved by linear programming optimally. We present a more computationally expensive approach to solve the problem, without the aforementioned assumptions. We also present numerical examples to show how this strategy can successfully mislead the $\chi^2$ failure detector. 

Our immediate future work is to study the game-theoretic aspects of the problem. In this work, we did not consider any active strategy being employed by the observer to mitigate the attack. In future works, we will consider the case of designing spoofing signals that explicitly take the attack mitigation strategies into account. In all the problems considered in this paper, the desired separations are taken as inputs provided by the user. The simulation results suggest that carefully choosing a specific profile of the desired separation can make it harder to detect by the observer. A possible extension is to automatically generate the optimal profile that not only minimizes the signal energy but also ensures that it is not detected by the observer. Another future work is to extend the strategy to more general non-linear state estimation approaches, such as the extended Kalman filter, unscented Kalman filter, and particle filters. 
\section*{APPENDIX}

\subsection{Proof of Theorem \ref{mk}}
\label{sec:proof_difference_mt}
Before we prove Theorem~\ref{mk}, we review the Kalman Filter update equations. Suppose the true measurement is $z_t$, the KF estimation is:
\begin{align}
\label{filter}
&{x}_{t|t-1} =  \mathcal{F}x_{t-1|t-1} +  \mathcal{G}u_t,\\
&{x}_{t|t}= \mathcal{F}{x}_{t|t-1}+K_t(z_t- \mathcal{H}{x}_{t|t-1}),
\end{align}
where $K_t$ is the Kalman gain and is given by:
\begin{equation}
\label{kalman_gain}
K_t = (\mathcal{F} {\Sigma}_{t|t-1} \mathcal{F}' + R_t)\mathcal{H}'(\mathcal{H}\Sigma_{t|t-1} \mathcal{H}' + Q_t )^{-1}.
\end{equation}

According to the Kalman gain update equation \eqref{kalman_gain}, the evolution covariance matrix at step $t$, $\Sigma_t$, only depends on the state model parameters and the initial condition of the covariance matrix $\Sigma_0$. The Kalman gain at step $t$, $K_t$ depends on the covariance matrix $\Sigma_t$. Both $\Sigma_t$ and $K_t$ do not depend on the control input series $\{u_t\}_{t=1,\cdots,k}$, measurement $\{z_t\}_{t=1,\cdots,k}$. Thus, the covariance matrix and the Kalman gain can be predicted from the KF covariance update steps.
\begin{equation} \label{Riccati}
\begin{split}
&{\Sigma}_{t+1|t}=\mathcal{F}{\Sigma}_{t|t}\mathcal{F}'+R_t,\\
&{\Sigma}_{t+1|t+1}= (I- K_t \mathcal{H}) {\Sigma}_{t+1|t}.
 \end{split}
\end{equation}

From Equation \eqref{Riccati}, the Kalman gain can be predicted from the initial condition $\Sigma_0$. 

We now prove our main result.

\begin{proof} 
From the update of  KF, we have
\begin{equation} \label{kalman_covariance}
\begin{split}
 m_t& = m_{t|t-1} + K_t(z_t-\mathcal{H}m_{t|t-1})\\
&=(I-K_t\mathcal{H})m_{t|t-1} +K_t z_t\\
&=(I- K_t\mathcal{H}) (\mathcal{F}m_{t-1}
+\mathcal{G}u_{t-1}) +K_t z_t.
\\
\end{split}
\end{equation}
and
$$\tilde{m}_t = (I- K_t\mathcal{H}) (\mathcal{F}m_{t-1} +\mathcal{G}u_{t-1}) +K_t (z_t + \epsilon_t).$$
Recursively,
\begin{equation}
\begin{split}
\label{Proposition2}
m_t & - \tilde{m}_t \\
=& (I- K_t\mathcal{H}) (\mathcal{F}m_{t-1} +\mathcal{G}u_{t-1}) +K_t z_t \\
&- [(I- \tilde{K}_t\mathcal{H}) (\mathcal{F}\tilde{m}_{t-1} +\mathcal{G}u_{t-1}) +\tilde{K}_t (z_t+\epsilon_t)]\\
=& (\mathcal{F} - K_t \mathcal{H} \mathcal{F})m_{t-1} - (\mathcal{F}- \tilde{K}_t \mathcal{H} \mathcal{F})\tilde{m}_{t-1}\\
&-(K_t - \tilde{K}_t)\mathcal{H} \mathcal{G} u_{t-1} + [K_t z_t - \tilde{K}_t(z_t+ \epsilon_t)]\\
=& (\mathcal{F} - \tilde{K}_t\mathcal{H}\mathcal{F})m_{t-1} - (\mathcal{F}- \tilde{K}_t\mathcal{H}\mathcal{F})m_{t-1}\\
&- (K_t - \tilde{K}_t) \mathcal{H} \mathcal{G} u_{t-1} +(K_t - \tilde{K}_t)z_t - \tilde{K}_t \epsilon_t\\
&-(K_t - \tilde{K}_t) \mathcal{H}\mathcal{F}m_{t-1}\\
=&(\mathcal{F} - \tilde{K}_t \mathcal{H}\mathcal{F})(m_{t-1} - \tilde{m}_{t-1})\\
&+(K_t - \tilde{K}_t)[z_t - \mathcal{H} (\mathcal{F}m_{t-1} + \mathcal{G}u_{t-1})] -\tilde{K}_t \epsilon_t.
\end{split}
\end{equation}

Define, $A_t = \mathcal{F} - \tilde{K}_t \mathcal{H}\mathcal{F}$, $B_t = (K_t-\tilde{K}_t)[z_t-\mathcal{H}(\mathcal{F}m_{t-1}+\mathcal{G}u_{t-1})]$ and $C_t = - \tilde{K}_t \epsilon_t.$
Then,
\begin{equation*}
\begin{split}
\label{Proposition2}
m_t  -& \tilde{m}_t \\
=& A_t (m_{t-1} - \tilde{m}_{t-1}) + B_t + C_t\\
=& A_t[A_{t-1}(m_{t-2} - \tilde{m}_{t-2}) + B_{t-1} + C_{t-1}] + B_t + C_t\\
\qquad \dots \\
=&\prod_{i=0}^{t-1}A_{t-i}\cdot (m_0  - \tilde{m}_0)+ (B_t + C_t)\\
&+ A_t (B_{t-1}+C_{t-1}) +  \cdots
+ A_t \cdots A_3 A_2 (B_1 + C_1)  \\
= &\prod_{i=0}^{t-1}A_{t-i}\cdot (m_0  - \tilde{m}_0) +B_t + C_t\\
&+\sum_{i=0}^{t-2} \left( \prod_{j=0}^{i}A_{t-j} (B_{t-1-i} + C_{t-1-i}) \right) .
\end{split}
\end{equation*}
\end{proof}

\bibliographystyle{IEEEtran}
\bibliography{IEEEabrv,main}
\begin{IEEEbiography}
[{\includegraphics[width=1in,height=1.25in,clip,keepaspectratio]{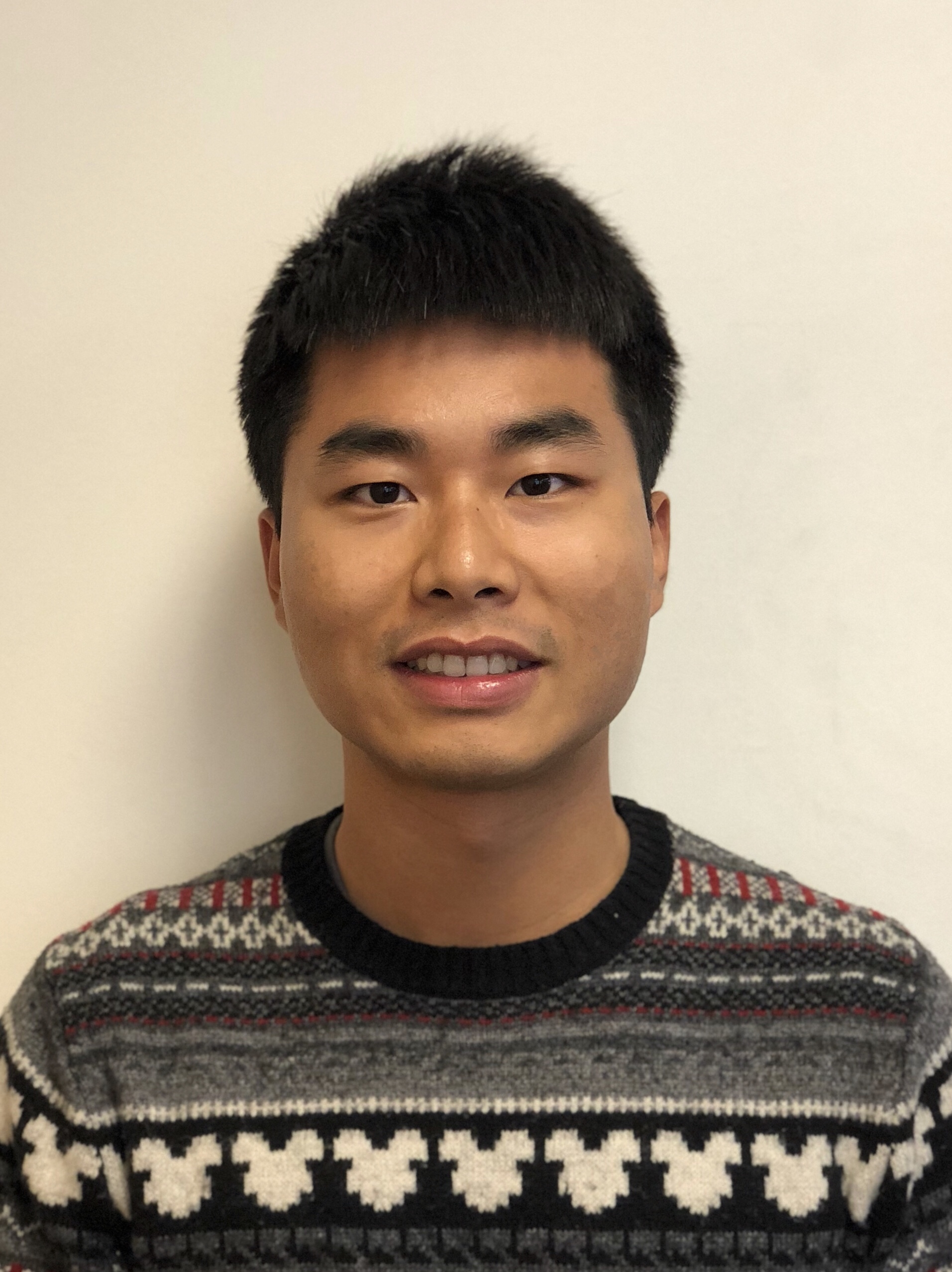}}]{Zhongshun Zhang}  received the B.S. degree in Electrical Engineering and Automation in 2012, and the M.Sc. degree in Control Engineering in 2015. Both from Southwest Jiaotong University, Chengdu, China. He is currently pursuing the Ph.D. degree in Electrical and Computer Engineering, Virginia Tech,  Blacksburg, VA, USA.

His research interests include pursuit-evasion games, state estimation, and target tracking algorithms.
\end{IEEEbiography}

\begin{IEEEbiography}
[{\includegraphics[width=1in,height=1.25in,clip,keepaspectratio]{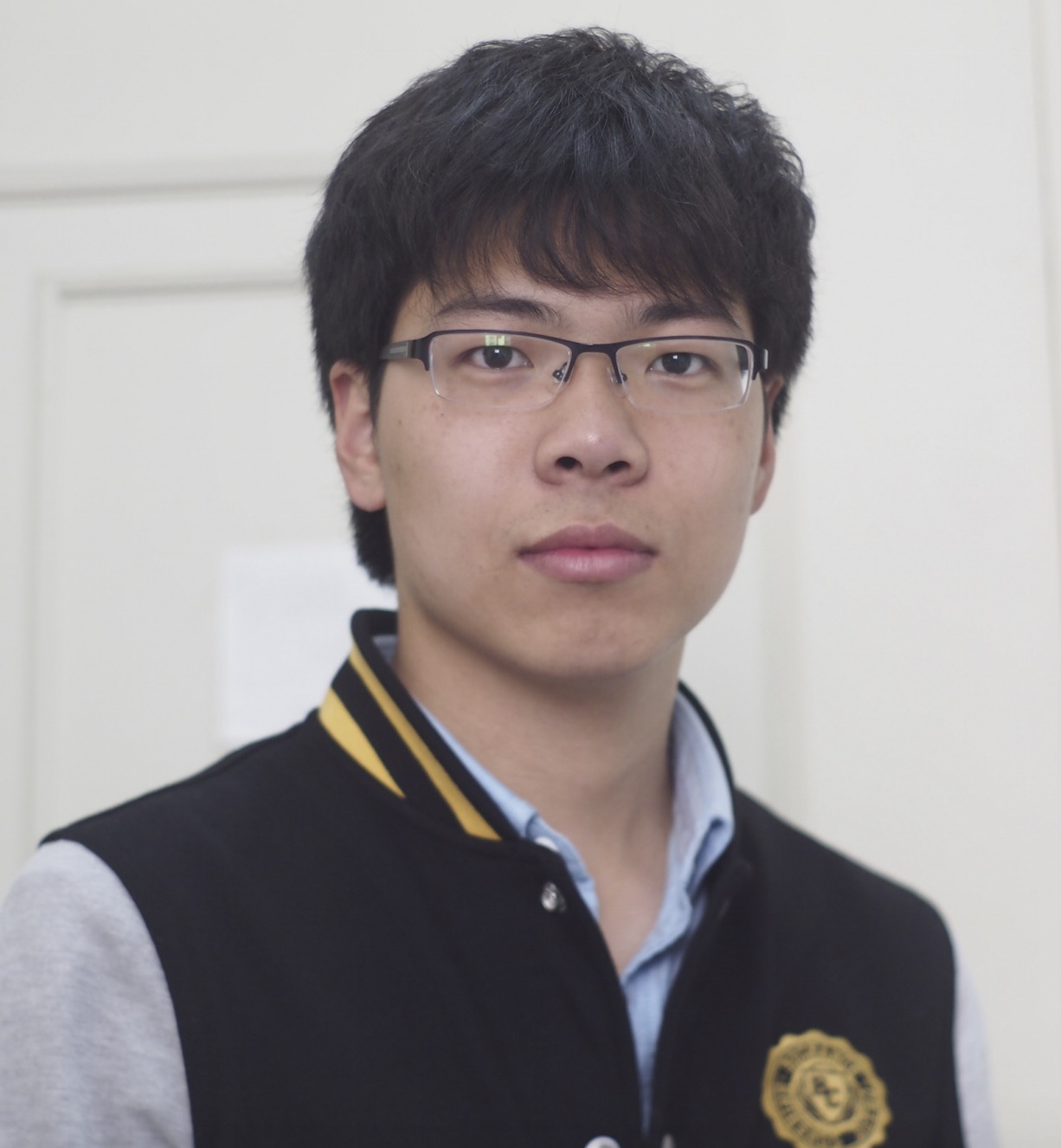}}]{Lifeng Zhou}  received the
B.S. degree in Automation from Huazhong University of Science and Technology, Wuhan, China, in 2013, the M.Sc.
degree in Automation from Shanghai Jiao Tong University, Shanghai, China, in 2016. He
is currently pursuing the Ph.D. degree in Electrical and Computer Engineering, Virginia Tech,
Blacksburg, VA, USA.

His research interests include multi-robot coordination, event-based control, sensor assignment, and risk-averse decision making.
\end{IEEEbiography}

\begin{IEEEbiography}[{\includegraphics[width=1in,height=1.25in,clip,keepaspectratio]{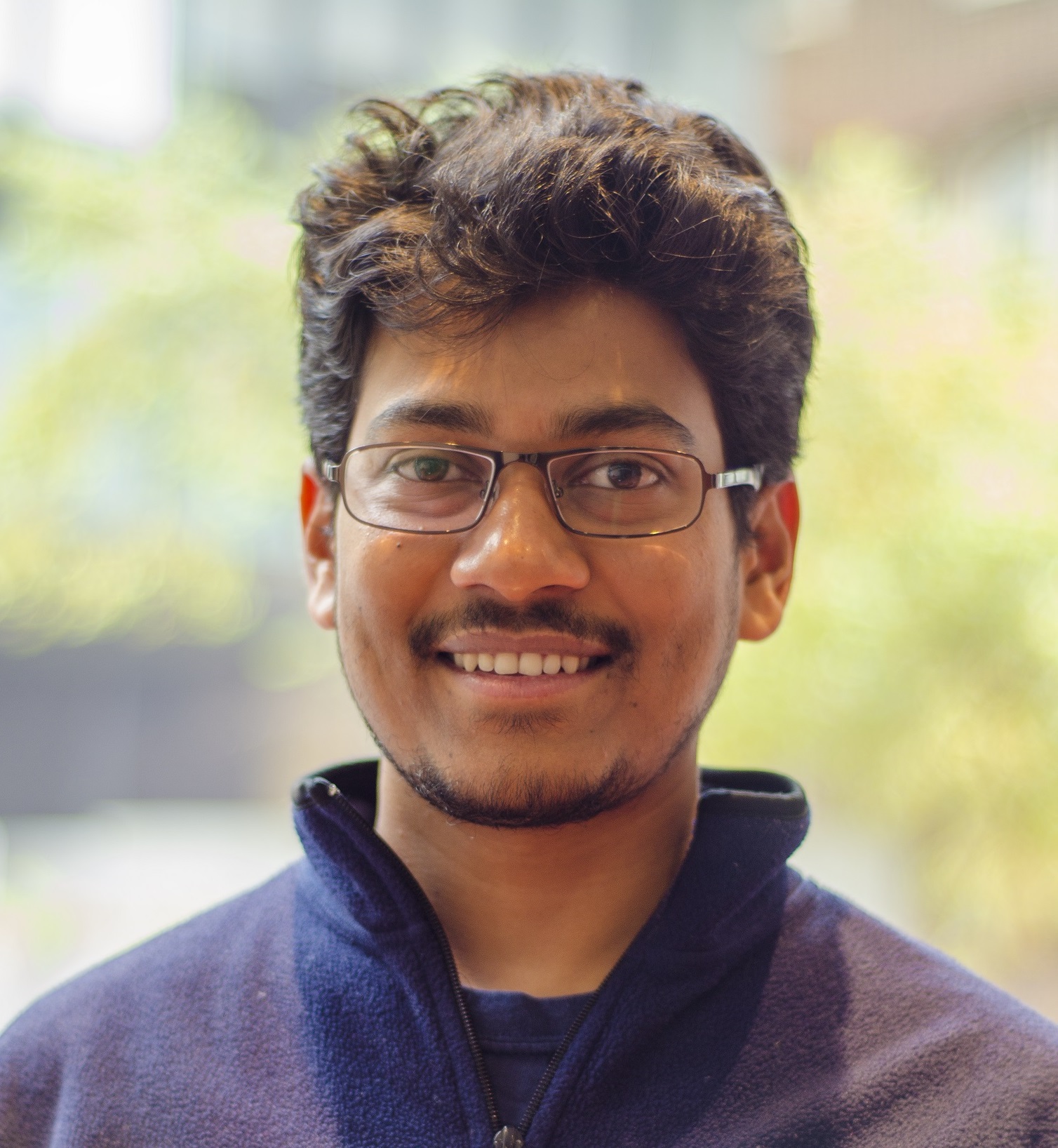}}]{Pratap Tokekar} is an Assistant Professor in the Department of Electrical and Computer Engineering at Virginia Tech. Previously, he was a Postdoctoral Researcher at the GRASP lab of University of Pennsylvania. He obtained his Ph.D. in Computer Science from the University of Minnesota in 2014 and Bachelor of Technology degree in Electronics and Telecommunication from College of Engineering Pune, India in 2008. He is a recipient of the NSF CISE Research Initiation Initiative award. His research interests include algorithmic and field robotics and  applications to precision agriculture and environmental monitoring.
\end{IEEEbiography}

\end{document}